\documentclass[journal]{IEEEtran}

\usepackage{graphicx}
\usepackage{epstopdf}
\usepackage[numbers,sort&compress]{natbib}
\usepackage{tikz}
\usetikzlibrary{shapes,arrows}
\usepackage{amsmath}
\usepackage{booktabs}
\usepackage{array}
\usepackage{bm}
\usepackage{relsize}
 \usepackage{algorithm} 
 \usepackage{algorithmic} 
 \usepackage{multirow} 
 \usepackage{amsmath}
 \usepackage{xcolor}

 \usepackage{color}
\usepackage{graphicx}

\usepackage{soul} 
\usepackage{amsfonts} 
\usepackage{amssymb} 
\usepackage{enumerate}
\usepackage{setspace}
\usepackage{microtype}
\usepackage{setspace}
\usepackage{amsthm}

\newtheorem{theorem}{Theorem}
\newtheorem{lemma}{Lemma}

\renewcommand{\algorithmicrequire}{\textbf{Initialization:}}
\renewcommand{\algorithmicensure}{\textbf{In the First Stage:}}
\renewcommand{\algorithmiclastcon}{\textbf{In the Second Stage:}} %

\begin{document}

\title{Efficient Hybrid Beamforming with Anti-Blockage Design for High-Speed Railway Communications}
\author{Meilin~Gao,~\IEEEmembership{Student Member,~IEEE},
        Bo~Ai,~\IEEEmembership{Senior Member,~IEEE},
        Yong~Niu,~\IEEEmembership{Member,~IEEE},\\
        Wen Wu,~\IEEEmembership{Member,~IEEE},
         Peng~Yang,~\IEEEmembership{Member,~IEEE},
         Feng~Lyu,~\IEEEmembership{Member,~IEEE},\\
         ~and~Xuemin~(Sherman)~Shen,~\IEEEmembership{Fellow,~IEEE}


\thanks{M. Gao, B. Ai, and Y. Niu are with the State Key Laboratory of Rail Traffic Control and Safety, Beijing Engineering Research Center of High-speed Railway Broadband Mobile Communications, and the School of Electronic and Information Engineering, Beijing Jiaotong University, Beijing 100044, P.R.China (e-mails:
\{meilingao, boai, niuyong\}@bjtu.edu.cn). (\emph{Corresponding author: Bo Ai.}) }

\thanks{W. Wu and X. Shen are with the University of Waterloo, Waterloo, ON N2L 3G1, Canada
(e-mails: \{w77wu, sshen\}@uwaterloo.ca).}

\thanks{P. Yang is with the School of Electronic Information and Communications, Huazhong University of Science and Technology, Wuhan, 430074, P.R.China (e-mail: yangpeng@hust.edu.cn).}

\thanks{F. Lyu is with the School of Computer Science and Engineering, Central South University, Changsha, 410083, P.R.China. (e-mail: fenglyu@csu.edu.cn).
}

\thanks{Part of this work has been presented at IEEE International Conference on Communications (ICC) 2019 \cite{Meilin2019ICC}.}}
\maketitle

\begin{abstract}
    Future railway is expected to accommodate both train operation services and passenger broadband services. The millimeter wave (mmWave) communication is a promising technology in providing multi-gigabit data rates to onboard users. However, mmWave communications suffer from severe propagation attenuation and vulnerability to blockage, which can be very challenging in high-speed railway (HSR) scenarios. In this paper, we investigate efficient hybrid beamforming (HBF) design for train-to-ground communications. First, we develop a two-stage HBF algorithm in blockage-free scenarios. In the first stage, the minimum mean square error method is adopted for optimal hybrid beamformer design with low complexity and fast convergence; in the second stage, the orthogonal matching pursuit method is utilized to approximately recover the analog and digital beamformers. Second, in blocked scenarios, we design an anti-blockage scheme by adaptively invoking the proposed HBF algorithm, which can efficiently deal with random blockages. Extensive simulation results are presented to show the sum rate performance of the proposed algorithms under various configurations, including transmission power, velocity of the train, blockage probability, etc. It is demonstrated that the proposed anti-blockage algorithm can improve the effective rate by 20$\%$ in severely-blocked scenarios while maintaining low outage probability.
\end{abstract}

\IEEEpeerreviewmaketitle


\section{Introduction}
 We have witnessed the rapid expansion of high-speed railway (HSR) transportation in the past decade, where railway communications are evolving at a fast pace to provide wireless broadband connections between road-side infrastructures and onboard travellers \cite{Ai2015future}.
 Innovation campaigns have been launched by railway operators, and the concept of ``smart rail'' is put forward by Shift2Rail \cite{SHIFT}. Smart rail services are envisioned to deliver consistent quality experiences, supporting a variety of services including autonomous driving, train multimedia dissemination, the Internet of Things for Railways (IoT-R), and onboard video surveillance \cite{Kato2019, Qiao2016Proactive, shen2020ai, guan2016millimeter}. Accommodating these data-craving applications in HSR scenarios is challenging to current railway communication systems, since the current narrowband Global System for Mobile Communications for Railways (GSM-R) has 9.6 kbps maximum transmission rate, and the available bandwidth of the ongoing Long Term Evolution for Railways (LTE-R) is limited to 20 MHz \cite{he2016high}. To satisfy these burgeoning demands on higher data rate, the promising millimeter wave (mmWave) system and multiple-input multiple-output (MIMO) technique, which are two typical technologies in fifth-generation (5G) era, are proposed to enhance the train-to-ground communications \cite{wu2019beef, ai2017indoor, liang2012efficient, Gao2018DynamicMB}. Furthermore, multi-user MIMO (MU-MIMO) technology is expected to reap multiplexing gains by coordinately scheduling multiple users on same spatial-temporal resources simultaneously \cite{Yan2016Position}. Owing to the plenteous mmWave spectrum bandwidth (from 30 GHz and 300 GHz), and magnificent spatial and diversity gains, the substantial increase in system capacity and spectrum efficiency can be achieved.

   Huge penetration loss and propagation attenuation are major barriers toward satisfied mmWave communication performance \cite{niu2015survey}. Fortunately, hybrid beamforming (HBF) and spatial multiplexing can be applied for further antenna and multiplexing gains via antenna arrays \cite{Dai2019H}.
   {\color{black}The correlation of mmWave sub-channels are exploited to reduce the computational complexity of hybrid beamforming \cite{yuan2018low}.}
   Plenty of existing works have targeted on low-mobility or static scenarios. In contrast, high-speed trains can move at 350 km/h or beyond, which results in distinct features such as short dwelling time and frequent handover \cite{Ai2014ChallengesTW}.
  {\color{black}More importantly, beamforming technique depends on the channel state information (CSI) which {\color{black}becomes} outdated quickly in high mobility scenarios, because the corresponding coherence time is much shorter in the rapidly time-varying channel.} Thus, developing efficient HBF techniques to address these challenges brought by high mobility in HSR communications is urgent.

 Furthermore, the high-directivity and short-wavelength make the mmWave beamforming sensitive to unexpected blockage events, when transmission signals propagate through objects such as buildings and vehicles \cite{Mhaske2016LinkQA, Lyu2019Cha}. Unfortunately, the blockage widely exists in both outdoor and indoor scenarios, which further brings the non-line-of-sight (NLoS) conditions into mmWave HSR scenarios, and blocks the line-of-sight (LoS) links between transceivers \cite{bababeik2017vulnerability}. Those random obstructions lead to intermittent connections, {\color{black}which adversely affect the stability and safety of moving high-speed trains, degrade the reliability and throughput, and deteriorate the experience of onboard passengers.}{\iffalse Eventually, it will restrict the potential of significant performance gain brought by mmWave communications \cite{Jain2018DrivenBC}. \fi} Extensive efforts have been devoted to combating the blockage in indoor environments. However, the developed approaches are costly and impractical to be applied in outdoor scenarios, considering the complicated reflectors and required dense back-up access points (APs) or base stations (BSs) deployment along the rail track. Besides, the dynamic outdoor blockages undermine the effectiveness of multi-path strategies. Few works study anti-blockage in outdoor mobility scenarios.
 Network densification and inter-BS handover are general measures against outdoor blockage \cite{Raghavan19Statistical}.
 Those methods are effective in conventional scenarios, however, in highly mobile HSR scenarios, wide distribution and long rail track {\color{black}require} high capital expenditure. Thus, it is also necessary to enhance the robustness against blockage in complicated HSR environments.

 {\iffalse In this paper, we investigate the hybrid beamforming for mmWave HSR communications to maximize the system sum rate in both blockage-free and blocked cases. The MU-MIMO system is deployed with multiple mobile relays (MR) on the top of the train to exploit the multiplexing gain and prevent penetration loss. \fi}

 {\color{black}In this paper, we consider {\color{black}the  mmWave MU-MIMO HSR communications}, with multiple mobile relays (MRs) on the top of the train to exploit the multiplexing gain and prevent the penetration loss. We investigate the efficient HBF design for mmWave HSR communications in both blockage-free and blocked scenarios.}
  In the blockage-free scenario, a two-stage HBF algorithm is proposed using the minimum mean square error (MMSE) approach for low complexity and fast convergence, and adopting the orthogonal matching pursuit (OMP) method for approximate beamformers recovery. In the blocked scenario, to achieve robust beamforming against blockage, we develop {\color{black}an anti-blockage beamforming scheme by intelligently invoking the proposed HBF algorithm,} to deal with the {\color{black}intermittent connections caused by random blockages.} Main contributions of this paper are summarized as follows.
    \begin{itemize}
    \item
    {\color{black}We formulate the sum rate maximization problem for the MU-MIMO HSR network with the optimal transceiver beamforming design, which is NP-hard and difficult to be solved directly.}
    \item
     We {\color{black}first} propose an efficient two-stage HBF algorithm with sum rate maximization in the blockage-free scenario. Specifically, in the first stage, the MMSE approach is adopted to achieve the optimal integrated hybrid beamformer design; while in the second stage, the analog and digital beamformers can be recovered from the optimal hybrid beamformer via the OMP approach {\color{black}approximately}.
    \item
    {\color{black}By detecting the link and capacity state}, the blocked scenario is categorized into three classes, then we take different strategies to {\color{black}invoke the HBF algorithm for better beamforming strategies against blockage}. We also analyze the computational complexity of the proposed algorithms, which is of polynomial-time complexity.
    \item
    We demonstrate that the proposed HBF algorithms can improve the system capacity significantly in both blockage-free and blocked scenarios, through extensive simulations. The outage probability is also largely reduced {\color{black}by the proposed anti-blockage algorithm. }
    \end{itemize}

    The remainder of this paper is organized as follows. We review the related work in Section II. Section III presents the system model and casts the problem formulation. In Section IV, we elaborate on the design of the HBF algorithm in the blockage-free scenario. Section V details the anti-blockage beamforming design and blockage performance metrics in the blocked scenario. We evaluate the algorithms by simulations in Section VI, followed by concluding remarks and future works in Section VII.

\section{Related Work}

\subsection{Hybrid Beamforming}
 There are significant efforts on the HBF for the augmented capacity of mmWave systems. To combat the huge propagation and penetration loss in mmWave systems, directional beamforming acts as a key enabler to enhance the link capacity and expand the transmission coverage by deploying large antenna arrays \cite{Kutty2016Beamforming, Yan2018StableBW}. To support multi-stream/multi-user demands and facilitate the tradeoff between the performance efficiency and hardware cost, HBF technique is proposed \cite{Molisch2017Hybrid, yu2018low, wu2019fast}. {\color{black}Robert} \textit{et al.} studied the hybrid precoding and codebook design with limited feedback for MU-MIMO mmWave systems in frequency selective channels \cite{Alkhateeb2016FrequencySH}.
 Independent transceiver design was assumed, indicating the beamformers and combiners were devised to maximize the mutual information at transceiver sides, respectively.

  {\color{black}With high-accuracy continuous localization and dynamic tracking of high-speed train, position information can be used for simplifying beam training process, Doppler shift compensation, and timing alignment \cite{talvitie2019positioning}.}
  To provide Internet access to onboard passengers in high-speed trains, Yan \textit{et al.} presented a position-based feedback scheme to support direct communications between track-side BS and onboard passengers \cite{Yan2016Position}. Passengers were firstly clustered into different sets according to their channel quality. Then, passengers in the same set were associated with the same feedback position and beamforming vector from the pre-assigned codebook. To alleviate the impact of the penetration loss and frequent handover, Song \textit{et al.} proposed some novel mmWave architectures with multiple MRs deployed on the train \cite{Song2016Millimeter}. The digital beamforming (DBF) and spatial multiplexing were integrated, which can support concurrent beam connections between multiple ground antenna arrays and multiple onboard antenna arrays.

 However, with the increase of large-scale mmWave antenna elements, the high hardware cost and energy consumption are inevitable for DBF in conventional HSR systems. Furthermore, the communication capacity fluctuates with the rapidly time-varying channel in HSR, which requires a robust beamforming scheme. Therefore, in this paper, we investigate the elaborated HBF design to balance the cost and efficiency, taking into account the fast time-varying multi-path channel and potential blocked scenarios.

 \subsection{Anti-blockage Approaches}
 In realistic mobile networks, {\color{black}wireless communication} links may be blocked by geographical/topographical blockages. The blockages can be divided into two classes, the self-blockage from user's hand or other body parts, and general blockages including other human bodies, vehicles, foliage, buildings, and other objects \cite{Raghavan19Statistical}. For a better understanding of the NLoS conditions rising from random blockage effects, various works have targeted on the blockage models and performance analysis \cite{Lyu2019Cha, bababeik2017vulnerability, Jain2018DrivenBC}. 
 To combat random blockages in practical scenarios, a variety of solutions have been proposed \cite{Jain2018DrivenBC, Yang2015, Wang2010ExploringMC, Niu2016ExploitingMR, Niu2019RelayAssistedAQ, Kim2017RelayassistedHT,  Zhang2012ImprovingNT, Ramrez2017OnOM}.

 In general, when dominant signal propagation paths between the transceivers are blocked by unexpected objects, the remedies are three-fold: the transmitter (TX) side, the receiver (RX) side, and intermediate links. From the point of deploying backup TX, {\textcolor[rgb]{0.00,0.00,0.00}{alternative}} APs/BSs can be densely deployed for substitute direct paths from new transmitters to the given RX, and to support diversity switching \cite{Jain2018DrivenBC, Zhang2012ImprovingNT, Ramrez2017OnOM}. In terms of the RX side, other receivers/nodes can be harnessed as a one-hop or multi-hop relay, which facilities connection restoration bypassing the obstacles \cite{Niu2019RelayAssistedAQ, Kim2017RelayassistedHT, Wang2010ExploringMC}. For intermediate links, it can turn to NLoS paths, transmission scheduling, and multi-hop routing, to exploit the spatial multiplexing gain \cite{Yang2015, Niu2016ExploitingMR}.

 Current approaches such as backup-BS, backup-AP, and multi-relay, are primarily suitable for short-range and indoor scenarios, which are costly and impractical to be implemented in outdoor scenarios. The approaches relying on NLoS paths, also suffer from severe attenuation and absorption. 
 {\color{black}To combat outdoor blockage, conventional network densification solutions require frequent handover and incur heavy signaling overhead, which are not suitable for HSR scenarios.}
 As an extension of our preliminary work \cite{Meilin2019ICC}, an anti-blockage beamforming strategy is proposed based on the HBF algorithm in blockage-free scenarios for HSR communications in this paper, with comprehensive simulation validations.

\section{System Model and Problem Formulation}

\subsection{\color{black}System Model}

\begin{figure}
\centering
\includegraphics[width=9cm]{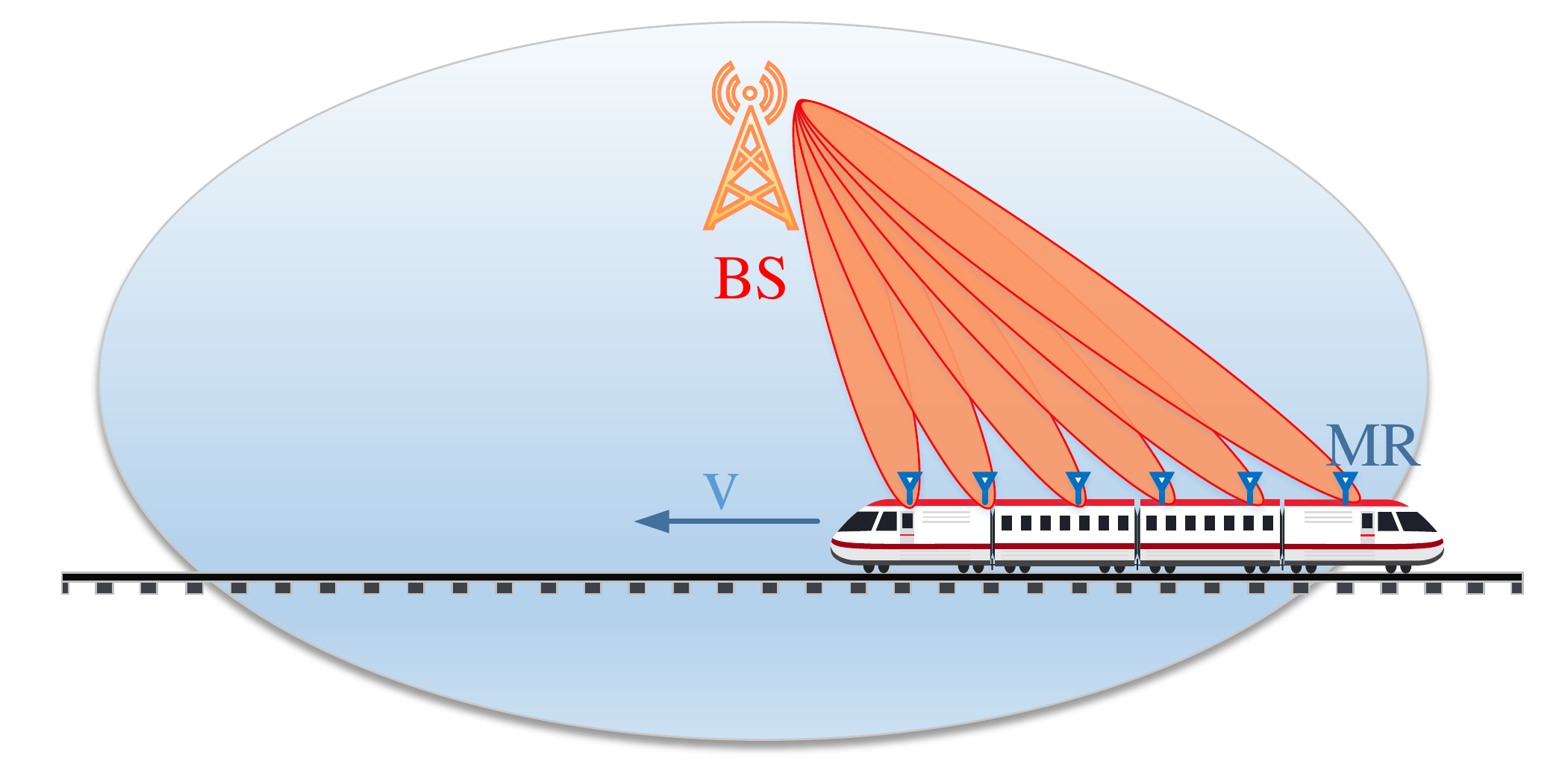}
\caption{The mmWave MU-MIMO system for HSR communications.}
\label{system_model}
\end{figure}

 As shown in Fig. \ref{system_model}, we adopt the mmWave MU-MIMO system model for HSR communications. Consider a single-cell scenario, where a train is traversing the coverage of the track-side BS. The track-side BS supports train-to-ground bidirectional communications via an array of antennas, which enables high data-rate links at mmWave frequency due to the proximity to the tracks.

 To overcome severe penetration loss of the train bodyshell, we consider a two-hop network that takes MRs as relays. On one hand, the MRs are deployed on the rooftop of the train and communicate with the BS through radio access links. On the other hand, the MR serves onboard users via the APs installed inside each carriage, to avoid the penetration loss and frequent handover. Note that, multiple radio access technologies can be enabled at the access links inside the train, including LTE/WiFi/3G, while the BS-MR links can be both at sub-6 GHz and millimeter wave frequencies. According to \cite{Zhang2015OptimalPA}, the links between the track-side BS and MRs are the main capacity bottleneck. Therefore, in this work, we resort to the mmWave BS-MR communications for higher capacity.

 The system performance may degrade when transmission links encounter random blockage events in practical HSR scenarios, such as the NLoS conditions caused by track-side obstructions including buildings, viaducts, trees, etc. By exploiting the spatial diversity gain, multiple MRs are mounted on the train, which enhances the system throughput and anti-blockage performance. All transceivers are equipped with multiple antennas, thus forming {\color{black}high directional communications, and we target on the downlink beamforming design}.

 \begin{figure*}
 \centering
 \includegraphics[width=12.0cm]{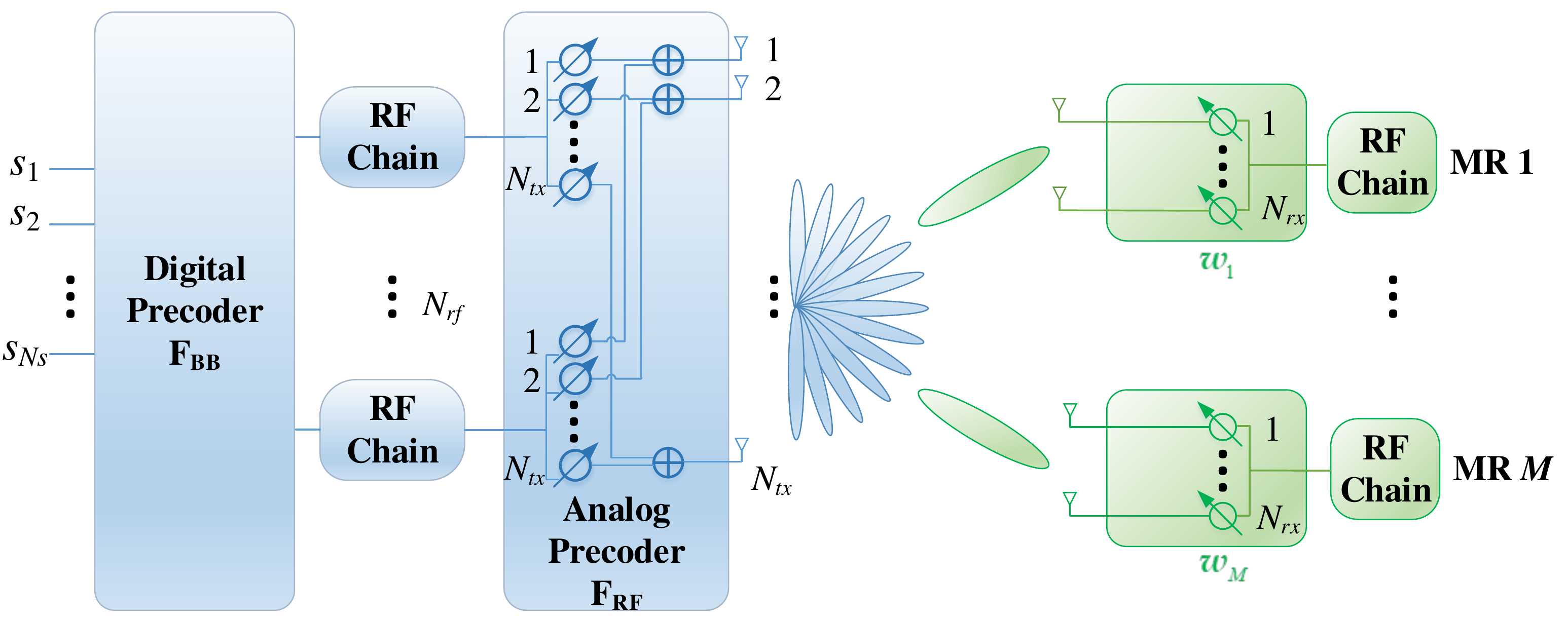}
 \caption{\color{black}Transceiver architecture of hybrid beamforming.}
 \label{hybrid_BF}
 \end{figure*}

\subsection{Hybrid Beamforming Structure}
 As shown in Fig. \ref{hybrid_BF}, we consider the HBF design consisting of low-complexity analog beamforming (ABF) and high-efficiency DBF.
 The HBF structure is implemented at the BS side with $N_{tx}$ antennas through $N_{rf}$ radio frequency (RF) chains, then delivers $N_s$ data streams to $M$ MRs. Without loss of generality, we assume that $N_s \leq N_{rf} \leq N_{tx}$, and $N_s$ is set to $M$ to fully exploit the channel gain. Generally, $M$ out of $N_{rf}$ RF chains {\color{black}are} chosen to serve $M$ MRs.

 At baseband (BB) processing, the symbol vector $\bm s =  [s_1,\ldots,s_{N_s}]^T $ corresponding to the MRs is firstly digitally precoded by the digital BB precoder $\mathbf{F}_{\rm BB}$, to suppress the inter-user interference, and $(\cdot)^T$ denotes the transpose operator. Afterward, RF chains up-convert the precoded signals to TX antennas via the analog RF beamformer matrix $\mathbf{F}_{\rm RF}$. ABF can be implemented with analog phase shifters, and reduce the hardware cost and power consumption in mixed signal components. To improve the system flexibility, a fully connected structure is adopted such that each RF chain can be connected to all TX antennas.

  The transmitted signal $\bm x$ at the TX side is represented as
  \begin{equation}
  \bm x=\mathbf{F}_{\rm RF}\mathbf{F}_{\rm BB}\bm s,
  \end{equation}
 where the analog beamformer can be expressed as $\mathbf{F}_{\rm RF} = [\bm {f}_{\rm RF}^1, \ldots, \bm {f}_{\rm RF}^{N_{rf}}]\in \mathbb{C}^{N_{tx} \times N_{rf}}$, the digital beamformer can be represented as $\mathbf{F}_{\rm BB} = [\bm f_{\rm BB}^1, \ldots, \bm {f}_{\rm BB}^{N_{s}}] \in \mathbb{C}^{N_{rf} \times N_s}$,
 and $\bm x \in \mathbb{C}^{N_{tx}}$. \textcolor[rgb]{0.00,0.00,0.00}{Generally, we assume $\mathbb{E}[\bm s \bm s^H] = \mathbf I_{N_s}$, where $\mathbb{E}[\cdot]$ denotes the expectation operator, and $(\cdot)^H$ denotes the conjugate transpose operator.}

 Due to the limited space and power supply at the RX side, each MR is equipped with an easy-deployable ABF precoder, including an RF chain and an array of $N_{rx}$ antennas. At each time slot, each MR is capable of receiving one data stream via directional beam from the BS, thus concurrent data transmission is supported.

  The processed signal at MR $m$, $\bm \hat{s}_m \in \mathbb{C}$, can be given by
  \begin{equation}\label{RX_sm}
  \begin{aligned}
  \bm \hat{s}_m = &{\bm w_m^H\mathbf{H}}_m \mathbf{F}_{\rm RF}\mathbf{F}_{\rm BB}\bm s + \bm w_m^H \bm n_m \\
  =&{\bm w_m^H\mathbf{H}}_m \sum_{i=1}^{M}\mathbf{F}_{\rm RF}\bm f_{\rm BB}^i s_i + \bm w_m^H \bm n_m,
  \end{aligned}
  \end{equation}
 where $\bm w_m \in \mathbb{C}^{N_{rx}}$ is the analog post-coding beamforming combiner at MR $m$, $\bm n_m \sim \mathcal{CN}(\mathbf{0}, \sigma_m^2 \mathbf I) \in \mathbb{C}^{N_{tx}} $ is a white Gaussian noise vector, and ${\mathbf{H}}_m \in \mathbb{C}^{N_{rx} \times N_{tx}}$ represents the channel matrix between the BS and MR $m$. Denote by $\mathbf{W} = \{\bm w_m\}_{m=1}^M$ the set of all receiver beamformers. {\color{black}Note that, $\bm s $ and $\bm n$ are uncorrelated random vectors.}

\subsection{Channel Model}
 The mmWave MIMO channel between transceiver antenna pairs is assumed to be frequency-selective. To incorporate the characteristics of both mmWave communications and HSR networks, a geometry-based stochastic channel model is adopted for 3GPP mmWave HSR scenarios, taking into account the spatial consistency and non-stationarity \cite{Yang2018A}. Given the frequency dependency, Yang \emph{et al.} model the wideband channel as the summation of the channel impulse response on separate sub-channels. The channel impulse response on each sub-channel is considered to be frequency-flat, and is a combination of the LoS path, reflected multi-path components, and scattered multi-path components. The complex amplitude on the sub-channel $k$ at time $t$ can be expressed as \cite{Yang2018A}
   \begin{equation}
   \begin{aligned}
    g_k(t)= & \sum_{l=1}^{L} g_{k,l}(t)= g_{k,LoS}(t) + g_{k,Ref}(t) + g_{k,Sca}(t)
    \\
  =  & \frac{c\cdot e^{-j 2 \pi f_{k}  \tau }}{4 \pi f_{k} d_{T,R}}
            + \sum_{p=1}^{N_{\rm Ref}} \frac{c \cdot  R(t) \cdot e^{j \phi_p}}{4 \pi f_{k} d_p}
         \\ &  + \sum_{q=1}^{N_{\rm Sca}} \frac{c \cdot S(t) \cdot e^{j \phi_q} }{4 \pi f_{k} d_{T,q}d_{R,q}} ,
   \end{aligned}
   \label{conv:Channel_gain}
  \end{equation}
 where $c$ denotes the speed of light, $f_{k}$ represents the carrier frequency at sub-channel $k$, $\tau  = d_{T,R}/c$ denotes the propagation delay of the LoS path, $R(t)$ and $S(t)$ indicate the link attenuation, $\phi_p$ and $\phi_q$ indicate the phase of corresponding reflected/scattered paths. Four distance terms including $d_{T,R}$, $d_{T,q}$, $d_{R,q}$ and $d_p$ stand for the corresponding distance between the transceiver, between the TX and the $q$-th scatter, between the RX and the $q$-th scatter, and the length of the $p$-th reflected path. It is assumed that there are $L$ paths, including $N_{\rm Ref}$ dominant reflective multi-paths and $N_{\rm Sca}$ effective scattered multi-paths, with $L = 1 + N_{Ref} + N_{Sca}$.

 The channel matrix ${\mathbf{H}}$ between the BS and RX is given by the summation of the channel matrix on each sub-channel ${\mathbf{h}_k}$
\textcolor[rgb]{0.00,0.00,0.00}{\begin{equation}
   \begin{aligned}
   {\mathbf{H}} &= \sqrt{\frac{N_{tx}N_{rx}}{L}} \sum_{k=1}^{K}\sum_{l=1}^{L} g_{k,l}(t){\bm a}_R(\varphi_l,\lambda_k) {\bm a}_T^H(\phi_l,\lambda_k),
   \end{aligned}
   \label{conv:Channel_matrix}
  \end{equation}}where the channel gain $g_{k,l}$ can be obtained from (\ref{conv:Channel_gain}),
   $\lambda_k$ is the wavelength of the sub-channel $k$ with $\lambda_k= c/{f_k}$, $\varphi_l$ and $\phi_l$ are the angles of arrival and departure of path $l$, ${\bm a}_T(\phi,\lambda)$ is the TX array response vector, which is given by
    \textcolor[rgb]{0.00,0.00,0.00}{ \begin{equation}
   \begin{aligned}
  {\bm a}_T(\phi,\lambda) = \frac{1}{\sqrt{N_{tx}}} [1, e^{j\frac{2\pi}{\lambda}\triangle_t \sin(\phi)},\ldots,e^{j(N_{tx}-1)\frac{2\pi}{\lambda}\triangle_t \sin(\phi)}]^T
     \end{aligned}
   \label{conv:a_T}
  \end{equation}}where $\triangle_t$ denotes the distance between TX antenna elements with $\triangle_t = \lambda/2$. The array response vector at the RX side, ${\bm a}_R(\varphi,\lambda)$, can be written in the same way.

   Given the processed signal at MR $m$ (\ref{RX_sm}) and the channel matrix (\ref{conv:Channel_matrix}), the achievable signal-to-interference-plus-noise ratio (SINR) can be written as
 \textcolor[rgb]{0.00,0.00,0.00}{\begin{equation}\label{conv:SINR}
   \gamma_m =  \frac{|\bm w_m^H {\mathbf{H}_{m} } \mathbf{F}_{\rm RF}\bm f_{\rm BB}^m s_m|^2}
{\sum_{i=1,i\neq m}^M  |\bm w_m^H {\mathbf{H}_{m} } \mathbf{F}_{\rm RF}\bm f_{\rm BB}^i s_i|^2 + \bm w_m^H \bm w_m \sigma_m^2},
 \end{equation}}where {\color{black}$\sigma_m^2$ is noise power}.
   {\color{black}Note that, the Doppler effect raised due to the high-mobility can be effectively tracked and compensated, based on the high-accuracy location information and regular movements of high-speed train \cite{Yan2016Position}.}

 From (\ref{conv:SINR}), the data rate with unit bandwidth at MR $m$ is given by the Shannon formula
 \begin{equation}\label{conv:rate}
 \begin{aligned}
    R_m  = \log_2\left( 1+ \gamma_m \right).
 \end{aligned} \end{equation}

\subsection{Problem Formulation}
 Considering the system capacity, we target on the transceiver beamformer/combiner design $\{\mathbf{F}_{\rm RF}, \mathbf{F}_{\rm BB},\mathbf{W}\}$, with the sum rate utility maximization. The optimization problem can be formulated as
 \begin{subequations}\label{conv:P1}
 \begin{align}
  \mathcal{P} 1:  \quad  \max_{\mathbf{F}_{\rm RF}, \mathbf{F}_{\rm BB},\mathbf{W}} \quad  &\sum_{m=1}^M R_m \label{P1_obj}\\
     \ \rm{s.t.} \qquad \ &
    {\rm{Tr}(\mathbf{F}_{\rm BB}^H \mathbf{F}_{\rm RF}^H \mathbf{F}_{\rm RF} \mathbf{F}_{\rm BB})} \leq P_0, \label{P1_cons}\\
    & \bm {f}_{\rm RF}^n \in \mathcal{D}, \quad n\in \{1,\ldots, N_{rf}\}
 \end{align}
 \end{subequations}
 where $P_0$ indicates the power budget at the BS, and operator $\rm{Tr}(\cdot)$ denotes the matrix trace operator. \textcolor[rgb]{0.00,0.00,0.00}{Generally, the analog beamformer can be realized based on a predefined codebook, \textcolor[rgb]{0.00,0.00,0.00}{$\mathcal{D}$, as standardized in IEEE 802.15.3c \cite{IEEE802.15}}. Essentially, the ABF design can be considered as selecting beams $\{\bm {f}_{\rm RF}^n\}_{n=1}^{N_{rf}}$ from the predefined codebook to compose the RF beamformer matrix $\mathbf{F}_{\rm RF}$.} The TX ABF/DBF beamformers can be {\color{black}considered} as an integrated hybrid beamformer $\mathbf{F} = \mathbf{F}_{\rm RF}\mathbf{F}_{\rm BB}$ with $\mathbf{F} \in \mathbb{C}^{N_{rx} \times N_{s}}$. For brevity, we assume $\bm f_m = \mathbf{F}_{\rm RF}\bm f_{\rm BB}^m$ and $\mathbf{F} = [\bm f_1, \ldots, \bm f_M]$. The constraint (\ref{P1_cons}) can be rewritten as
 \begin{equation*}
   \rm{Tr}(\mathbf{F}^H \mathbf{F}) \leq P_0.
 \end{equation*}

 Problem $\mathcal{P} 1$ is challenging to be solved due to the following reasons. First, $\mathcal{P} 1$ is not a convex problem as the objective function is nonconvex logarithmic, and the constraints are quadratic. Meanwhile, the coupling of ABF and DBF further makes the objective function and power constraints complicated. In addition, the optimal RF beam steering is constrained by the predefined beam codebook via analog phase-shifters. Consequently, this optimization problem is NP-hard and difficult to be solved directly. In the following section, we propose a two-stage transceiver precoders/combiners design to solve problem $\mathcal{P} 1$.

\section{Two-Stage Hybrid Beamforming Design\\ in Blockage-Free Cases}
 In this section, the blockage-free scenario is considered, and a two-stage HBF algorithm for sum rate maximization in the mmWave HSR scenario is proposed. Firstly, the analog beamformer is tackled together with the digital beamformer as an integrated hybrid beamformer $\mathbf{F}$, which can be obtained by the MMSE algorithm. 
 In the second stage of the ABF/DBF matrix design, i.e., $\mathbf{F}_{\rm RF}$ and $\mathbf{F}_{\rm BB}$, are approximately recovered from $\mathbf{F}$ via the OMP approach.

 \subsection{First Stage: MMSE-based Transceiver Beamforming}
 In this stage, we start from the design of TX hybrid beamformer $\mathbf{F}$ and RX combiners $\mathbf{W}$, to maximize the sum rate utility of the mmWave BS-MR links.

 Problem $\mathcal{P}1$ is intractable mainly due to the logarithmic objective function. Targeting on the nonconvex sum rate maximization, we transform this problem into an equivalent problem sharing the {\color{black}same optimal solution}.
 {\color{black}Inspired by \cite{lin2019hybrid}, we adopt the weighted minimum mean square error method which can facilitate the problem solution with integrating the sum rate maximization problem and the spectral efficiency maximization problem.}
 In particular, we resort to a linear beamformer approach with the mean square estimation error (MSE) estimator \cite{Shi2011An, Wang2017SpectrumAE}.
 Based on (\ref{RX_sm}) {\color{black}and the independence of $\bm s$ and $\bm n$}, the MSE at MR $m$ is given by
  \begin{equation}\label{e_m}
    \begin{aligned}
       e_m \triangleq &\mathbb{E}[(\hat{s}_m - s_m )^H (\hat{s}_m - s_m)]\\
       = & \mathbb{E}[
        ({\bm w_m^H\mathbf{H}}_m \sum_{i=1}^{M}\mathbf{F}_{\rm RF}\bm f_{\rm BB}^i s_i + \bm w_m^H \bm n_m  - s_m )^H \\
      &\quad \cdot ({\bm w_m^H\mathbf{H}}_m \sum_{i=1}^{M}\mathbf{F}_{\rm RF}\bm f_{\rm BB}^i s_i + \bm w_m^H \bm n_m- s_m)
       ]
       \\
       = & | 1- \bm w_m^H {\mathbf{H}_m } \bm{f}_m| ^2
      + \bm w_m^H \bm w_m \sigma_m^2
      + \sum_{i = 1, i\neq m}^{M}|\bm w_m^H {\mathbf{H}_m } \bm{f}_i| ^2.
    \end{aligned}
 \end{equation}
 \begin{theorem}
    The original problem $\mathcal{P}1$ is equivalent to the following problem $\mathcal{P}2$ with the same optimal beamformers:
    \begin{subequations}
        \begin{align}
         \mathcal{P}2: \quad  \max_{\mathbf{F},\mathbf{W}} \quad  &\sum_{m=1}^M -\log_2(e_m) \\
          \rm{s.t.}  \quad & \rm{Tr}(\mathbf{F}^H \mathbf{F}) \leq P_0,\\
           & \bm {f}_{\rm RF}^n \in \mathcal{D}, \quad n\in \{1,\ldots, N_{rf}\}.
        \end{align}
    \end{subequations}
 \end{theorem}

\begin{proof}
 We first formulate an MMSE problem with an objective of minimizing the sum MSE of all receivers, i.e.,
 \begin{equation}
 \begin{aligned}
     \min_{\mathbf{F},\mathbf{W}} \quad  &\sum_{m=1}^M e_m
     \\
      \rm{s.t.} \quad&
         \rm{Tr}(\mathbf{F}^H \mathbf{F}) \leq P_0,\\
    & \bm {f}_{\rm RF}^n \in \mathcal{D}, \quad n\in \{1,\ldots, N_{rf}\}.
 \end{aligned}
 \label{conv:Perror}
 \end{equation}
 Fixing the TX hybrid beamformer $\mathbf{F}$, we can obtain the optimal RX combiner at MR $m$ by solving
 \begin{equation}\label{arg_min_em}
   \bm w_m^* = \arg \min_{\bm w_m} \ \  \sum_{m=1}^M e_m.
 \end{equation}
 According to the first-order optimality condition $\frac{\partial e_m}{\partial \bm w_m} = \mathbf 0$ \cite{boyd2004convex}, the optimal combiner at the MR $m$ is obtained as
  \begin{equation}\label{w_m}
    \bm w_m^* = \frac{\mathbf{H}_m \bm{f}_m}{\sum_{i = 1}^{M} \|\mathbf{H}_m\bm{f}_i\|^2_2 + \sigma_m^2}.
\end{equation}
Substituting (\ref{w_m}) into (\ref{e_m}), we can obtain the optimal MSE as
 \begin{equation}\label{e_m*}
 \begin{aligned}
    e_m^*     =  & 1 - \frac
    {
    \|\mathbf{H}_m \bm{f}_m\|^2
    }
    {
    \sum_{i = 1}^{M}  \|\mathbf{H}_m\bm{f}_i\|^2_2 + \sigma_m^2
    }\\
    = & \frac
    {
           \sum_{i = 1, i \neq m}^{M}  \|\mathbf{H}_m\bm{f}_i\|^2_2 + \sigma_m^2
    }
    {
      \sum_{i = 1}^{M}  \|\mathbf{H}_m\bm{f}_i\|^2_2 + \sigma_m^2
    }    .
    \end{aligned}
 \end{equation}
 Based on (\ref{w_m}), the achievable SINR at MR $m$ (\ref{conv:SINR}) is written as
  \begin{equation}\label{SINR_opt}
 \begin{aligned}
   \gamma_m^* =  \frac
    {
     \|\mathbf{H}_m\bm{f}_m\|^2_2
    }
    {
      \sum_{i = 1, i \neq m}^{M}  \|\mathbf{H}_m\bm{f}_i\|^2_2 + \sigma_m^2
    },
 \end{aligned}
 \end{equation}
 and the data rate in (\ref{conv:rate}) can be rewritten as
 \begin{equation}\label{R_e}
 \begin{aligned}
    R_m^* = & \log_2\left( 1+
      \frac
    {
     \|\mathbf{H}_m\bm{f}_m\|^2_2
    }
    {
      \sum_{i = 1, i \neq m}^{M}  \|\mathbf{H}_m\bm{f}_i\|^2_2 + \sigma_m^2
    }
     \right) \\
     = &\log_2\left(
    \frac
    {
     \sum_{i = 1}^{M}  \|\mathbf{H}_m\bm{f}_i\|^2_2 + \sigma_m^2
    }
    {
      \sum_{i = 1, i \neq m}^{M}  \|\mathbf{H}_m\bm{f}_i\|^2_2 + \sigma_m^2
    }
 \right)
 \\
    = & \log_2\left( \frac{1}{e_m^{*}}\right)
     =  -\log_2\left(e_m^{*}\right).
 \end{aligned}
 \end{equation}
    Given the same constraints and the equivalency of objective functions as derived above, the original problem $\mathcal{P}1$ and problem $\mathcal{P}2$ are proved to be equivalent.
\end{proof}

To solve the optimal TX hybrid beamformer, we first give the following Lemma.
\begin{lemma}
   The problem $\mathcal{P}2$ is equivalent to the following problem $\mathcal{P}3$ with the same optimal beamformers:
\begin{equation}
   \begin{aligned}
      \mathcal{P} 3: \quad \max_{\mathbf{F}, \mathbf{W},\bm{\alpha}} \quad
      & \sum_{m=1}^M \log_2(\alpha_m) - \frac{\alpha_m e_m}{\ln2}\\
      \ \rm{s.t.} \quad \
      &
              \rm{Tr}(\mathbf{F}^H \mathbf{F}) \leq P_0,\\
    & \bm {f}_{\rm RF}^n \in \mathcal{D}, \quad n\in \{1,\ldots, N_{rf}\}.
        \end{aligned}
 \label{conv:P2}
 \end{equation}
\end{lemma}

\begin{proof}
  To tackle the difficult logarithmic optimization problem $\mathcal{P}2$, we introduce an auxiliary function as
  \begin{equation*}
    \mathcal{U}(\alpha_m) = \log_2(\alpha_m) - \frac{\alpha_m e_m}{\ln2} + \frac{1}{\ln2},
  \end{equation*}
   which is a function of the auxiliary variable $\alpha_m$. Let $\bm{\alpha}=$ $[\alpha_1,\ldots,\alpha_M]$ denote the positive weight vector over MSE indicator $e_m$. We take the first derivative of $\mathcal{U}(\alpha_m)$ over $(\alpha_m)$ and let it to be zero. Then we can obtain the optimal solution $\alpha_m^* = \frac{1}{e_m}$ and the maximum $\mathcal{U}^{*}$ is
    \begin{equation}\label{U_opt}
      \mathcal{U}^* = -\log_2(e_m),
    \end{equation}
    according to the first-order optimality condition \cite{boyd2004convex}.
    Therefore, $\mathcal{P}2$ is proved to be equivalent to the optimization problem $\mathcal{P}3$, in the sense that the optimal beamformer solution $\mathbf{F}$ and $\mathbf{W}$, for the two problems are identical under the auxiliary variable $\alpha_m^* = \frac{1}{e_m}$.
\end{proof}

\begin{theorem}
 The optimal TX hybrid beamformer of problem $\mathcal{P}1$ can be obtained by
    \begin{equation}
     \bm f_m^* =
     \left(
        \sum_{m = 1}^M \mathbf{H}_m^H \bm w_m \alpha_m \bm w_m^H \mathbf{H}_m
        + \lambda_m^* \bm{I}_{N_{tx}}
     \right)^\dag
     \mathbf{H}_m^H \bm w_m \alpha_m,
     \label{conv:f_m}
     \end{equation}
     where $(\cdot)^\dag$ is the pseudo-inverse operator, and $\lambda_m^{*}$ is the optimal Lagrange multiplier.
  \end{theorem}

\begin{proof}
 Given the equivalence of $\mathcal{P}2$ and $\mathcal{P}3$, the optimal RX combiner $\bm w_m^*$ given by (\ref{w_m}) also satisfies $\mathcal{P}3$.
  Recall from Lemma 1 that the maximum value of $\mathcal{P}3$ lies on the point $\alpha_m^{*} = \frac{1}{e_{m}}$.
  To determine the hybrid beamformer $\mathbf{F}$, the {\color{black}Lagrangian function} associated with problem $\mathcal{P}3$ is defined as
    \begin{equation}
    \begin{aligned}
        \mathcal{L}(\mathbf{F},\lambda) = \sum_{m=1}^M \log_2(\alpha_m) - \frac{\alpha_m e_m}{\ln2} + \lambda \left(
         \rm Tr(\mathbf{F}^H \mathbf{F})
         - P_0 \right).
    \end{aligned}
    \end{equation}

\begin{algorithm}
\caption{Two-Stage Hybrid Beamforming Algorithm 
} \label{alg:the hybrid BF algorithm}
\begin{algorithmic}[1]
\REQUIRE ~~   \\
  Acquire the channel state information ${ \{\mathbf{H}_m\}_{m=1}^{M}} $.\\
  Initialize TX beamformer $\bf F$ with
  $\rm{Tr}(\mathbf{F}^H \mathbf{F}) \leq  {P_{0}}$.\\
  Adopt the codebook $\mathcal{D}$.
\ENSURE $//$ {\textit{MMSE}}(${\bf  H}$, $\mathcal{D}$)  ~~\\
\WHILE {$ \sum_{m=1}^M e_m^* \geq  \epsilon$}
\STATE
    Update RX beamformer ${\bm{w}_m}\!^*$ by (\ref {w_m})\\
\STATE
    Update MMSE coefficient $e_m^* $ by (\ref {e_m*})\\
\STATE
    Update TX hybrid beamformer $ {\bm f}_m^*$ by (\ref {conv:f_m}) \\
\ENDWHILE\\
\STATE
Output: $\!{\bf W}^*\!=\![{ {\bm w}_1}\!^*,..., { {{\bm w}}_M}\!   ^*]$, $\!{\bf  F}^*\! =\! [{\bm f}_1^*,...,  {\bm f}_M^*]$
\LASTCON $//$ {\textit{OMP}}(${\bf  F}^*$, $\mathcal{D}$) ~~
\STATE
    ${\bf  F}_{res} = {\bf  F^*}$, ${\bf F_{\rm RF} = {\rm \Phi}}$, $\bf F_{\rm BB} = {\rm \Phi}$\\
\FOR {$n = 1:{N_{rf}}$}
\STATE
    Select the most correlated beam vector $\bm d_i$ by (\ref {conv:d_i})
\STATE    Recover RF beamformer ${\bf F^*_{\rm RF}}(:,n) = {\bm d}_i$\\ \vspace{1mm}
\STATE    Recover BB beamformer ${\bf F^*_{\rm BB}} = {\bf F_{\rm RF}^\dag} {\bf  F}_{res}$\\
\STATE    Update the residual matrix by ${\bf  F}_{res} = {\bf  F}_{res} \!-\! \mathbf{F}^*_{\rm RF} \mathbf{F}^*_{\rm BB}$\\
\STATE    Update the codebook with ${\mathcal{D}}(:,j)= \bm 0$\\
\ENDFOR\\
\STATE
$
{\bf F}^*_{\rm BB} = \sqrt{P_0} \ {\bf F}^*_{\rm BB} / \|{\bf F}^*_{\rm RF}{\bf F}^*_{\rm BB} \|_2^2$, $\mathbf{F}_{recov} = {\bf F}^*_{\rm RF}{\bf F}^*_{\rm BB}
$
\STATE
Output: $\mathbf{F}_{recov}$, ${\bf F}^*_{\rm RF}$, ${\bf F}^*_{\rm BB}$.
\end{algorithmic}
\end{algorithm}

 The corresponding Karush-Kuhn-Tucker (KKT) conditions \cite{boyd2004convex} are shown as follows
 \begin{subequations}
 \begin{align}
 & \frac{\partial}{\partial{\bm{f}_m}}\mathcal{L}(\mathbf{F},\lambda) = 0, \quad m = 1,\ldots, M,\\
 & \rm{Tr}(\mathbf{F}^H \mathbf{F}) - P_0 \leq 0,\\
 & \lambda \geq 0,\\
 & \lambda \left( \rm{Tr}(\mathbf{F}^H \mathbf{F}) - P_0 \right) = 0,
 \end{align}
 \end{subequations}
 where $\lambda$ is the Lagrange multiplier associated with the inequality constraint.

The hybrid beamformer at the TX side can be solved with the KKT conditions \cite{Shi2011An}, yielding (\ref{conv:f_m}).
The iterative subgradient method can be applied to update the Lagrange multiplier $\lambda_m$
  \begin{equation}
\lambda_m^{t+1} = \lambda_m^{t} - \xi\left( \rm{Tr}(\mathbf{F}^H \mathbf{F}) - P_0 \right) ,
  \end{equation}
 where $\xi$ indicates the step size. Given the equivalence of the three optimization problems, $\mathcal{P}1$, $\mathcal{P}2$, and $\mathcal{P}3$, the optimal TX beamformers and RX combiners obtained in $\mathcal{P}3$ also satisfy $\mathcal{P}1$. Hence proved.
\end{proof}

 Now the optimal TX beamformer and RX combiner, i.e., $\mathbf{F}^*$ and $\mathbf{W}^*$, satisfying the original optimization problem $\mathcal{P}1$ have been solved in a low-complexity manner. Next, ABF/DBF beamformers are designed based on the TX hybrid beamformer and a given codebook.

 \subsection{Second Stage: OMP-based TX Beamformer Decoupling}
 In this stage, the ABF/DBF beamformers at the BS side, ${\bf F}_{\rm RF}$ and ${\bf F}_{\rm BB}$, are decoupled from above hybrid beamformer ${\bf F^*}$, \textcolor[rgb]{0.00,0.00,0.00}{whereby the ABF precoder ${\bf F}_{\rm RF}$ is realized based on a codebook for low complexity. {\color{black}Suppose the predefined codebook $\mathcal{D}$ with $D$ beam patterns, each implemented with one normalization vector.} Inspired by the compressed signal representation, the iterative greedy OMP method is adopted \cite{Alkhateeb2013HybridPF}.}

 Our target is to reconstruct the ABF matrix from the codebook $\mathcal{D}$, to approximate the optimal HBF design. In these settings, each ABF steering vector, which corresponds to one column in the matrix of $\mathbf{F}_{\rm RF}$, is chosen from the codebook matrix (called ``Dictionary''). Each column vector in this Dictionary, $\mathcal{D}(:,j)$, is termed as ``atom''. The core idea is to reconstruct a sparse matrix by iteratively choosing a dictionary atom based on the principle which is mostly correlated with the current residuals. To achieve fast convergence, the residual signals always yield the orthogonal projection of the subspace generated by all already selected atoms.

 The process of OMP can be summarized as follows. During initialization, current residual $\mathbf F_{res}$ is set to be the optimal hybrid beamformer $\mathbf{F}^*$. In each iteration, a new atom $\bm d_i \in \mathcal{D}$ is appended as a beamforming column vector in $\mathbf{F}_{\rm {RF}}$, which is the most correlated with current residual as \begin{equation}
 \bm d_i = {\arg \max \limits_{\mathcal{D}(:,j)}|\langle\mathcal{D}(:,j), \mathbf F_{res}\rangle|},
 \label{conv:d_i}
 \end{equation}
 {\color{black}where $\langle \cdot,\cdot \rangle$ is the inner product operator.}
 The residual is updated as the orthogonal projection of the subspace of all already selected atoms. Note that, the beamforming vectors in codebook $\mathcal{D}$ cannot be chosen repeatedly. By the end of all iterations, we can get separate ABF and DBF as ${\bf F}^*_{\rm RF}$ and ${\bf F}^*_{\rm BB}$, and the recovered ${\bf F}_{recov}$, based on optimal hybrid beamformer ${\bf F}^*$ and the predefined codebook $\mathcal{D}$.

 In Algorithm 1, the developed two-stage HBF scheme in the ideal blockage-free scenario for mmWave HSR systems is elaborated. \textcolor[rgb]{0.00,0.00,0.00}{Based on Algorithm 1, all TX beamformers and RX combiners can be obtained in the ideal blockage-free HSR scenarios. The two-stage hybrid beamformer implementation yields the achievable SINR $\{\gamma_{m}\}_{m=1}^M$, and the system capacity as $C = \sum_{m=1}^{M}R_m$.} Besides, the fast-fading time-varying channel in HSR can be captured and addressed by the MMSE method, while the OMP stage can recover the optimal hybrid beamformer approximately.

 \section{Anti-blockage Beamforming Design}
Considering random blockages in practical scenarios may hinder the dominant radio links between the BS and MRs, we investigate the anti-blockage beamforming design. Based on the proposed HBF algorithm, we present {\color{black}an anti-blockage beamforming strategy} in practical multi-path/LoS blocked scenarios, and study the impact of blockage on sum rate under various transmission power, MR settings and blockage probabilities.

\subsection{Anti-blockage Algorithm}
 \begin{algorithm}
 \caption{Anti-blockage Beamforming Algorithm}\label{alg:anti-blocakge BF algorithm}
 \begin{algorithmic}[1]
\REQUIRE Obtain hybrid beamformers from Algorithm 1: \{$\mathbf{F}^*$, ${\bf F^*_{\rm RF}}$, ${\bf F^*_{\rm BB}}$, ${\bf W^*}$\}; ${\bf  F}_{res} = {\bf F^*}$ \\
\STATE \textit{Detect the link state} $\gamma$ between the BS and MRs\\
\IF {some link is blocked with $\gamma_m < \gamma_{th}$ }
\STATE \textit{Detect the capacity state} $C$ of the system\\
\IF { $C < C_{th}$ }
\STATE $//$ \textit{update the MMSE stage}
\STATE ${\mathbf{H}}_m = \bm 0$\\
\STATE $(\bf  F_{bl}^*,\bf  W_{bl}^*) =$ {\textit{MMSE}}(${\bf  H}$, $\mathcal{D}$)\\
\STATE ${\bf  F}_{res} = {\bf F_{bl}^*}$\\
\ENDIF
\FOR {$n = 1:{N_{rf}}$}
\STATE $//$ \textit{update the OMP stage}
\IF{$\gamma_n < \gamma_{th}$}
\STATE ${\bf F^*_{\rm RF}}(:,n) = \bm 0$\\
\ELSE {}
\STATE Select the beam steering vector $\bm d_i$ by (\ref {conv:d_i})
\STATE Recover RF beamformer ${\bf F^*_{\rm RF}}(:,n) = {\bm d}_i$\\
\STATE Recover BB beamformer ${\bf F^*_{\rm BB}} = {\bf F_{\rm RF}^\dag} {\bf  F}_{res}$\\
\STATE Update the codebook with ${\mathcal{D}}(:,j)= \bm 0$\\
\ENDIF
\ENDFOR\\
\STATE  $ {\bf F}^*_{\rm BB} = \sqrt{P_0} \ {\bf F}^*_{\rm BB} / \|{\bf F}^*_{\rm RF}{\bf F}^*_{\rm BB} \|_2^2$, $\mathbf{F}_{recov} = {\bf F}^*_{\rm RF}{\bf F}^*_{\rm BB} $
\ENDIF \\
Output: $\mathbf{F}_{recov}$, ${\bf F}^*_{\rm RF}$, ${\bf F}^*_{\rm BB}$
\end{algorithmic}
\end{algorithm}

 Due to the short wavelength, mmWave is vulnerable to various blockages, such as foliage, buildings, and viaducts in railway environments. The link blockage occurs when obstacles appear in the radio links between the transceivers, which results in {\color{black}received signal {\color{black}strength} degradation caused by severe attenuation}. When the achieved SINR at the RX side is {\color{black}lower} than the required threshold $\gamma_{th}$, the system is unable to guarantee the required bit error rate and the link is considered to be ``turned off'' \cite{Kutty2016Beamforming}. The link blockage depends on multiple factors, including the surrounding environment, obstacle density, beamwidth, and {\color{black}transmission distance}. For analytical simplicity, it is assumed that the link blockage probability remains stable in a road segment. The link blockage probability $p_b$ on average is deemed as a constant in a certain section along the rail track \cite{Yang2015}.

\textcolor[rgb]{0.00,0.00,0.00}{Theoretically, by detecting both the link and capacity state, the blockage conditions can be identified and further categorized. First of all, the system detects the CSI when the train enters the cell coverage. The CSI is obtained at the RX side on the train, which is further shared to the TX side, i.e. BS, via the feedback mechanism. Based on the detected CSI, the SINR at the RX side can be calculated in following time slots. The link state is periodically detected on the received SINR within each time slot.}

  {\color{black}With both the link and capacity state detections, the blocked scenarios are categorized into three different classes. If all radio links are detected with desired achievable SINR $\gamma_m \geq \gamma_{th}$, it indicates the system works ideally or in the slightly-blocked scenario (\textit{Class I}); otherwise, further detection on the capacity state is required. If the capacity $C$ is detected to be above the threshold $C_{th}$, the mildly-blocked scenario happens (\textit{Class II}); otherwise, deteriorated capacity ($C \leq C_{th}$) indicates the severely-blocked scenario (\textit{Class III}).}

 By incorporating the blockage effects, an anti-blockage beamforming strategy is proposed by intelligently modifying the HBF design in ideal blockage-free scenarios. {\color{black}To fit different levels of blockage conditions, we adopt accordingly strategies as follows}
{\color{black}\begin{itemize}
  \item For Class I under the slightly-blocked scenario, to achieve a balance between the complexity and performance efficiency, it is unnecessary to update the HBF design in Algorithm 1.
  \item For Class II under the mildly-blocked scenario, the blockage may inflict an intermediate link condition, where the deteriorated link quality is lower than the desired threshold $\gamma_{th}$ but the system capacity is still acceptable. 
   In this case, we can turn to the beamforming redesign, and modify the OMP stage by turning off the blocked links.
  \item For Class III under the severely-blocked scenario, significant blockage attenuation hampers the radio transmission between transceivers, resulting in a low transmission capacity. 
      In this case, the CSI of the blocked links should be reset to zero. The whole HBF algorithm with both MMSE and OMP stage should be re-executed based on the updated CSI. Once the achieved link SINR and system capacity both stay below the threshold, we should remeasure the CSI of all radio links.
\end{itemize}}

  {\color{black}As noted, only when the anti-blockage beamforming algorithm fails to prevent the link deterioration lower than the desired link quality,
    it is necessary to remeasure the CSI. In this way, it can significantly reduce the channel detection overheads while guaranteeing the rate performance, by avoiding frequently periodical detection on CSI within each time slot. Hence, the complexity in the anti-blockage algorithm is significantly reduced.}

 To evaluate the anti-blockage performance, we first need to define the criteria for link outage. The outage probability experienced by MR $m$ is defined as the probability that the achievable SINR is below a certain threshold, which can be {\color{black}obtained} by
 \begin{equation}\label{out_link}
   P_{out,m} = \frac{1}{Q} \sum_{\iota = 1}^{Q} {\bm 1}_{\{\gamma_m^\iota \geq \gamma_{th}\}},
 \end{equation}
 where $Q$ denotes the number of total trials, and $\gamma_m^\iota$ indicates the received SINR at user $m$ in the $\iota$-th trial. The indicator function ${\bm 1}_{\{\rm condition\}}$ equals one if the condition is satisfied, and equals zero otherwise.

 Accordingly, the system outage probability can be given by
 \begin{equation}\label{p_out}
  P_{out} = \frac{1}{M}\sum_{m=1}^{M}P_{out,m}.
 \end{equation}
 The average blocked channel capacity can be derived as
 \begin{equation}\label{out_cap}
  C_{block} = \sum_{m=1}^{M} R_m P_{out,m}.
 \end{equation}
  From (\ref{out_cap}), the blocked capacity is the attenuated version of the free space capacity, affected by random blockage events.

 \subsection{Complexity Analysis}
 The computational complexity analysis of the proposed HBF algorithm and anti-blockage beamforming algorithm is crucial and necessary. {\color{black}The proposed algorithms in blockage-free and blocked scenarios in the worst case (i.e., in the severely-blocked scenario), both including two stages for hybrid beamformer design. Thus the computational complexity of the two algorithms is at the same level.} Suppose the analysis is under the condition of given TX/RX antenna elements deployment.

 In the first stage, the complexity per iteration of the MMSE method comes from the calculation of RX MMSE beamformer $\bm w_m^*$, the MMSE coefficient $e_m^*$, and the TX hybrid beamformer $\bm f_m^*$. To note, the linear MMSE method yields the complexity of $\mathcal{O}(M)$ with the output of $\bm w_m^*$, $e_m^*$. Furthermore, to solve the Lagrangian function for optimizing $\bm f_m^*$, the bisection method is applied to update the Lagrange multiplier $\lambda_m^*$ with the complexity of $\mathcal{O}(M^2\log_2(\frac{1}{\epsilon}))$ \cite{Shi2011An}, where {\color{black}the termination criteria} $\epsilon$ determines the accuracy of the Algorithm 1. For the second stage to reconstruct separate ABF/DBF beamformers, the complexity of the OMP method is $\mathcal{O}(M \ln N_{tx})$ \cite{Alkhateeb2013HybridPF}. Therefore, the two-stage algorithm yields the per-iteration complexity of $\mathcal{O}(M^2\log_2(\frac{1}{\epsilon})+M + M \ln N_{tx})$, which is of polynomial-time computational complexity.

\section{Performance Evaluation}

 In this section, we evaluate the performance of the proposed HBF algorithms in blockage-free and blocked scenarios with various simulation configurations. Specifically, in comparison with other typical beamforming algorithms, we measure the system performance mainly on sum data rate and outage probability. Moreover, the {\color{black}impact} of critical system parameters including the transmission power, the velocity of the train, and the number of MRs {\color{black}is} also investigated. We consider a single-cell where a train traverses at a constant speed of $v$. Main simulation configurations are listed in Table I.

 \begin{table}[tbp]\small
 \centering  
 \caption{Simulation Parameters}
 \begin{tabular}{lccc}  
 \hline
        Parameter &\!Symbol &Value \\ \hline  
         Carrier frequency &\emph{f}$_c$ &32 GHz\\
        System bandwidth &\emph{W} &500 MHz \\         
        Noise power density &$N_0$ &--174 dBm/Hz \\        
        Height of BS &\emph{h}$_{\emph{BS}}$ &10 m \\
        Height of MR &\emph{h}$_{\emph{MR}}$ &2.5 m \\
        BS coverage radius  &\emph{R} &600 m \\
        Number of TX antennas&\emph{N}$_{tx}$ & 32  \\
        Number of RX antennas  &\emph{N}$_{rx}$ & 16 \\
        Number of TX RF chains  &\emph{N}$_{rf}$ & 8 \\
        SINR threshold & $\gamma_{th}$    & 10 dB \cite{Kutty2016Beamforming} \\
        Termination criteria & $\epsilon$ & $10^{-3}$ \\
        {\color{black}Number of paths} & $L$ & $5$ \\
        {\color{black}Number of dominant reflective paths\!} &$N_{\rm Ref}$ & $2$ \\
        {\color{black}Number of effective scattered paths\!} &$N_{\rm Sca}$ & $2$ \\
           \hline
 \end{tabular}
 \end{table}

 \subsection{Benchmark Schemes and Evaluation Metrics}

\begin{figure*} \centering
 \begin{minipage}[t]{0.32\textwidth}
 \centering
  \includegraphics[scale=0.4]{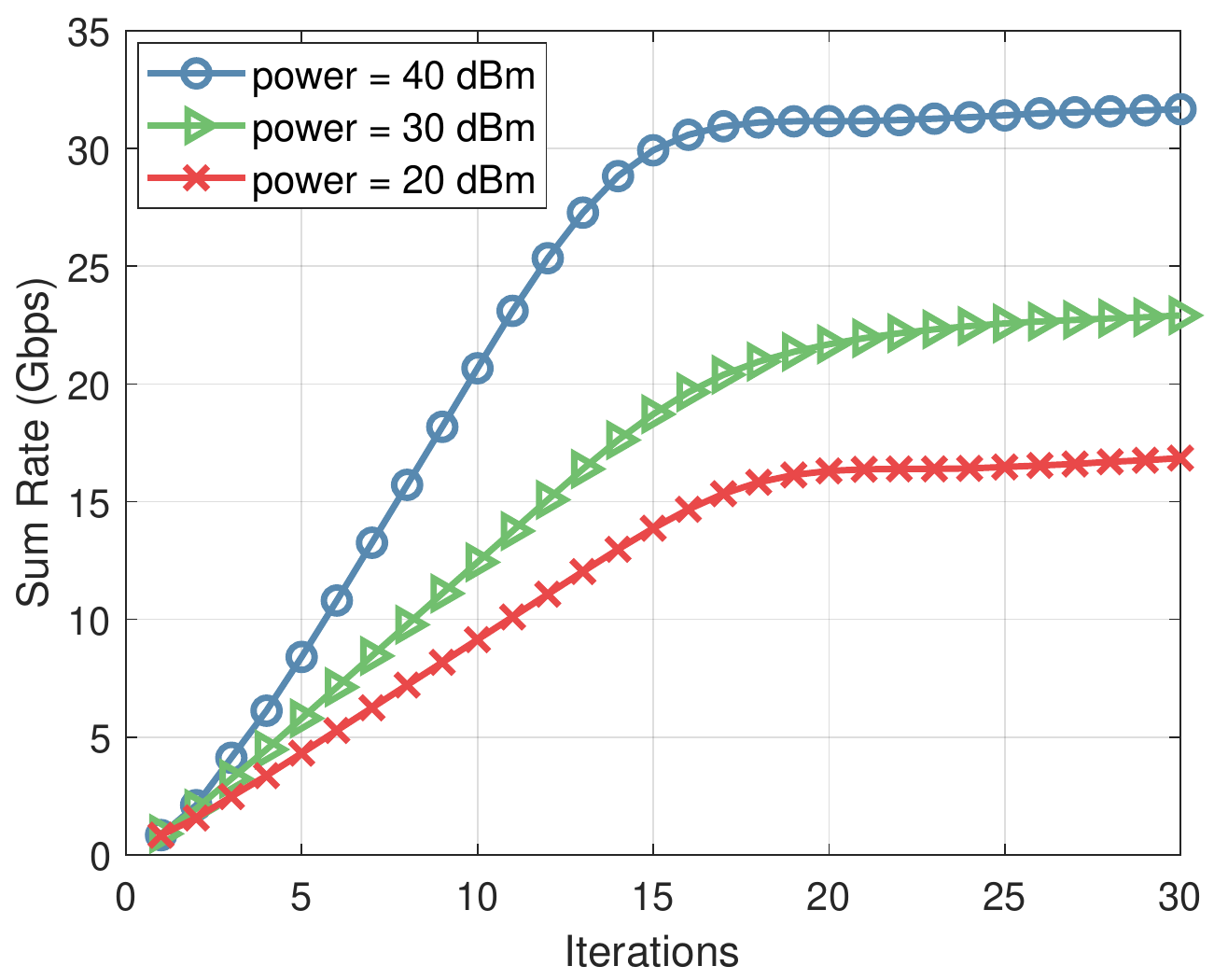}
  \caption{Convergence performance of Algorithm 1.}
  \label{conv:simu_iter}
\end{minipage}
\hspace{0.05in}
 \begin{minipage}[t]{0.31\textwidth}  \centering
  \includegraphics[scale=0.31]{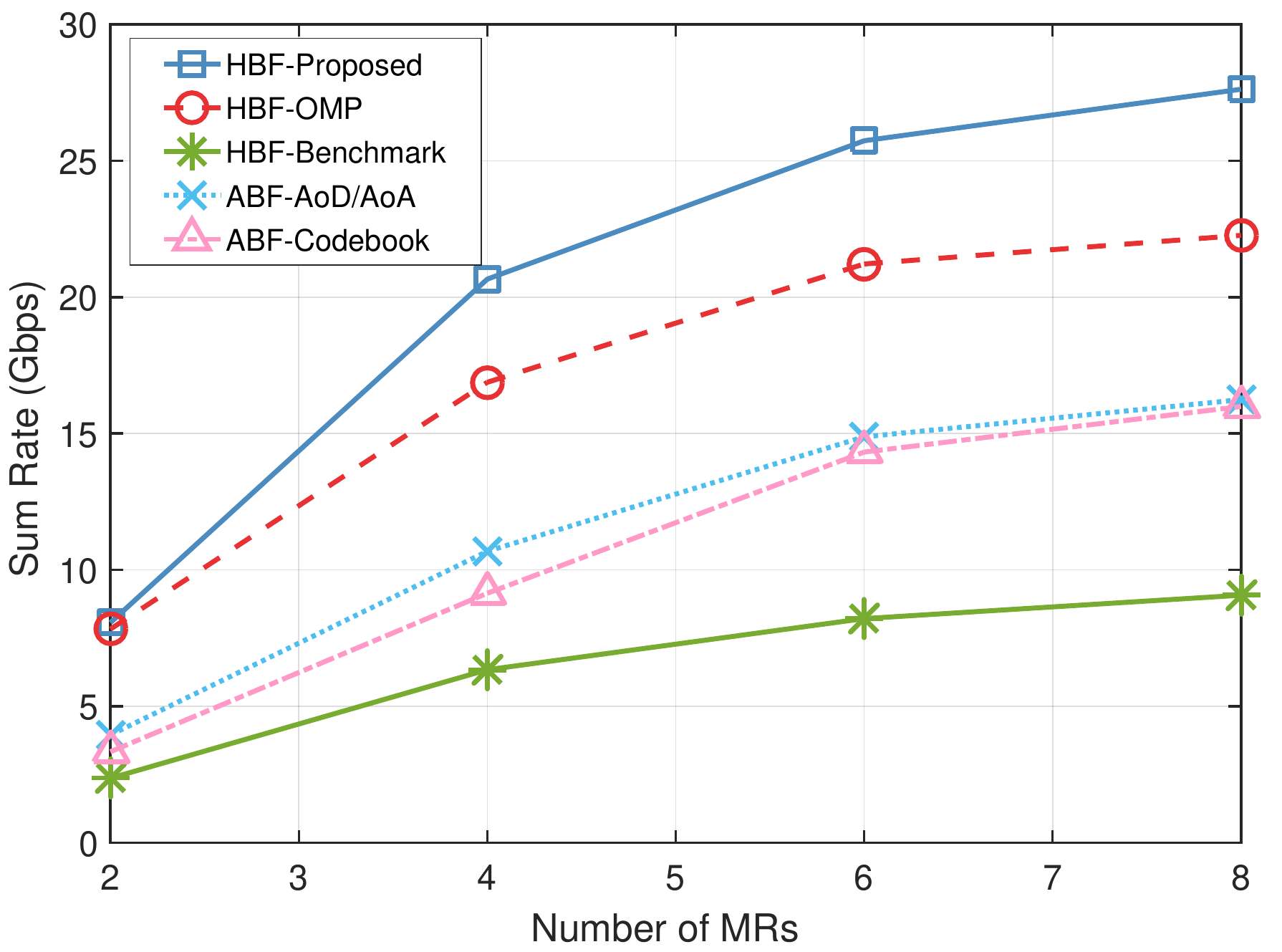}
  \caption{Sum rate vs. the number of MRs.}
  \label{Fig_Rate_MR}
\end{minipage}
\hspace{0.05in}
  \begin{minipage}[t]{0.31\textwidth}\centering
  \includegraphics[scale=0.31]{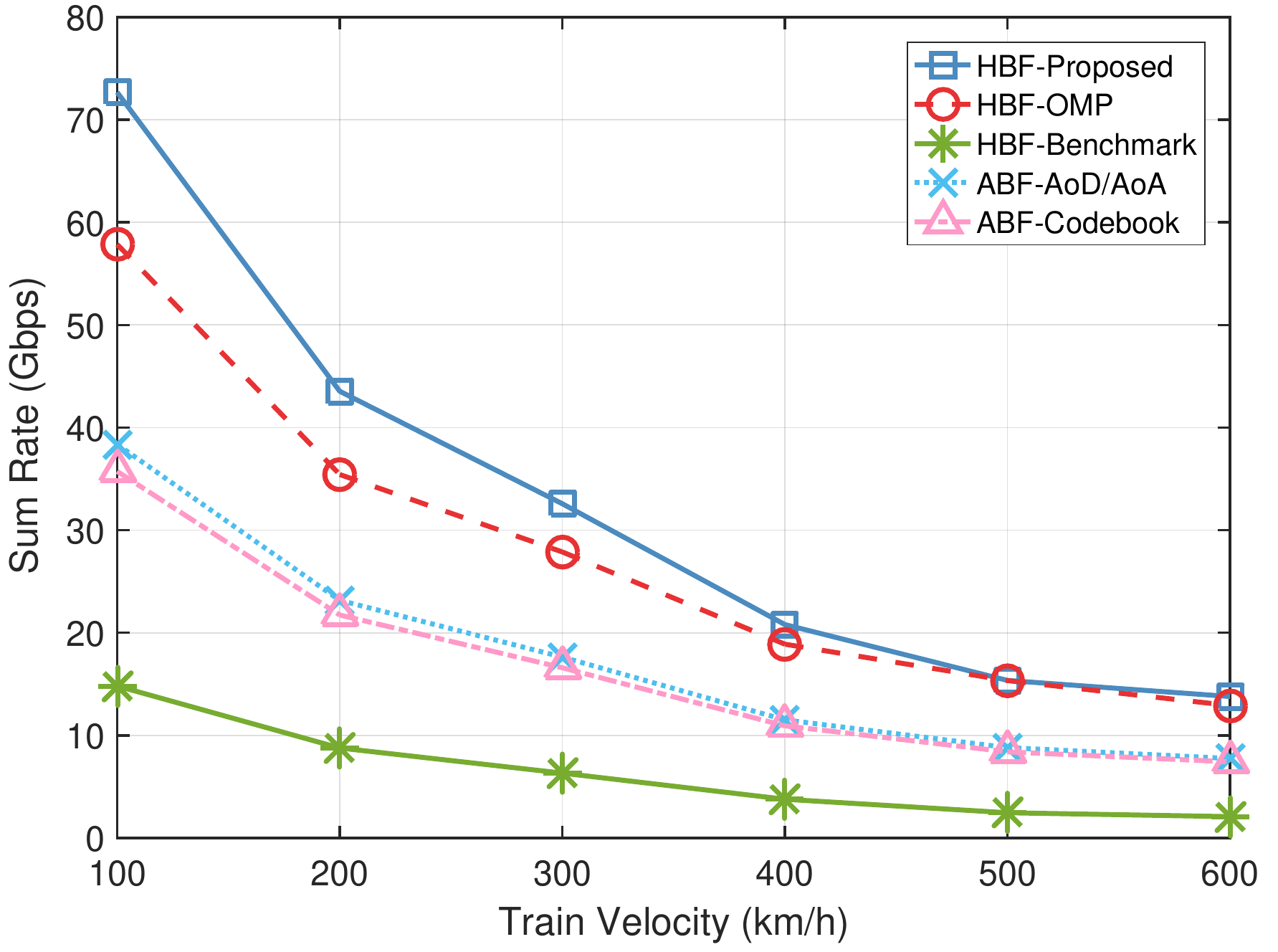}
  \caption{Sum rate vs. the velocity of the train.}
  \label{Fig_Rate_Speed}
  \end{minipage}
\end{figure*}
\textcolor[rgb]{0.00,0.00,0.00}{To evaluate performance }of the proposed HBF (HBF-Proposed) algorithm, another four beamforming schemes are provided for comparison:
 \begin{enumerate}[1)]
    \item \textbf{\textit{HBF-Benchmark}}: In this HBF algorithm, the analog/digital beamformers at the BS side are randomly set with satisfying the predefined codebook constraints and the transmission power budget. The receiver analog combiners are also randomly selected from {\color{black}predefined codebooks}.
    \item \textbf{\textit{HBF-OMP}}: In the proposed two-stage HBF algorithm, the first stage of MMSE hybrid beamformer design makes a tremendous contribution to the complexity reduction while achieving the suboptimal capacity performance. Therefore, the HBF-OMP algorithm is executed to {\color{black}emphasize} the significance of the MMSE stage, with only the second stage OMP of the HBF-Proposed algorithm based on the randomly preset hybrid beamformer.

    \item \textbf{\textit{ABF-AoD/AoA}}: In comparison with above HBF structures, {\color{black}the ABF scheme is performed}. Beam tracking is generally performed by tracking the angle of departure/arrival (AoD/AoA) of the dominant paths, thus adapting the directions of the beams. At the cost of frequent channel estimation and beam alignment overheads, the ABF-AoD/AoA can also achieve a good performance \cite{Gao2018DynamicMB}.
    \item \textbf{\textit{ABF-Codebook}}: Codebook-based beam switching is another typical ABF scheme and is superior to the beam tracking in terms of complexity and overhead, by sweeping finite beamspace over a codebook \cite{IEEE802.15}.
 \end{enumerate}

\begin{figure*}
\centering
\begin{minipage}[t]{0.49\textwidth}
  \centering
  \includegraphics[width=0.75\columnwidth,height=50mm]{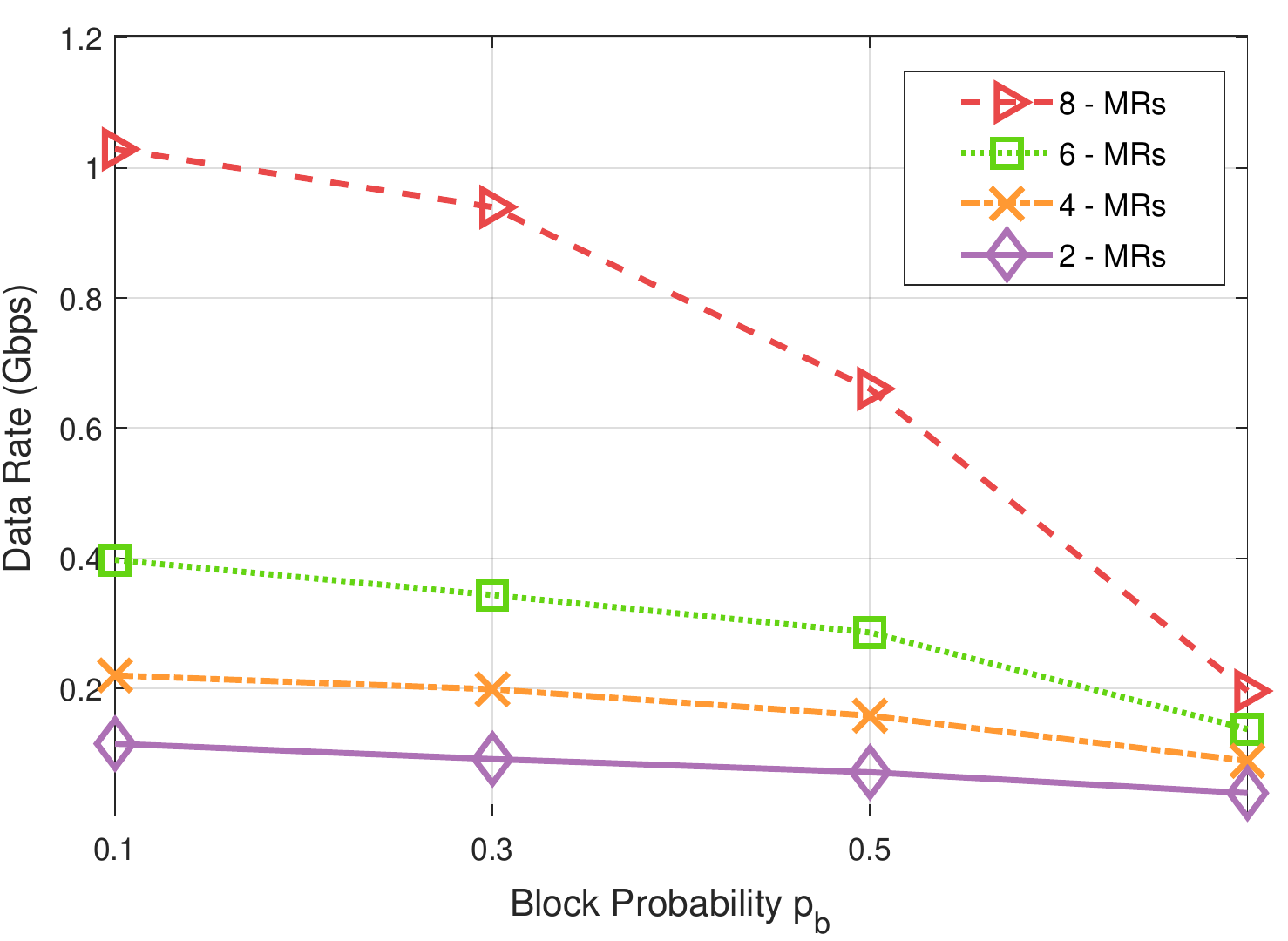}
  \caption{Data rate vs. blockage probability under different numbers of MRs.}
  \label{Fig_Rate_Block}\vspace*{-0mm}
\end{minipage}
\hspace{0.05in}
\begin{minipage}[t]{0.49\textwidth}
  \centering
  \includegraphics[width=0.75\columnwidth]{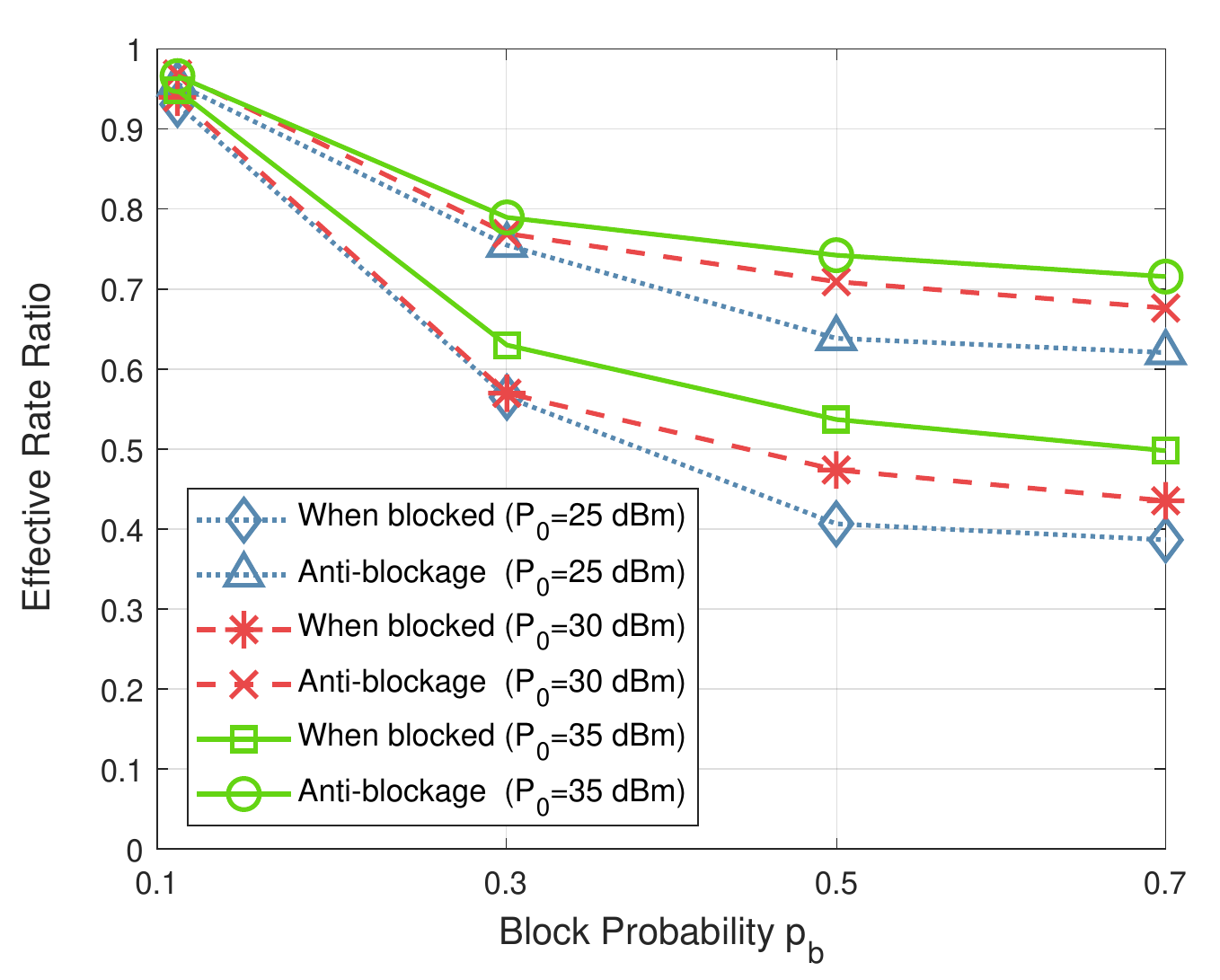}
  \caption{Effective rate ratio vs. blockage probability \textcolor[rgb]{0.00,0.00,0.00}{under different conditions}.}
  \label{Fig_Ratio_Cases_Block}\vspace*{-0mm}
\end{minipage}
\end{figure*}
 Two main performance metrics are evaluated, i.e., data rate and outage probability. Specifically, three kinds of data rate are considered, i.e., the instantaneous \emph{data rate} achieved by all MRs per one time slot, the \textit{sum rate} of the system {\color{black}with system bandwidth} inside the cell coverage, and the \emph{blocked sum rate} with {\color{black}random} blockage calculated by Eq. (\ref{out_cap}). To better demonstrate the impact of the blockage, the \emph{effective rate ratio} of the blocked sum rate to the system sum rate without blockage is defined as $\frac{C_{block}}{C}$, instead of the blocked sum rate. {\color{black}The outage probability in (\ref{p_out}) is measured to verify the outage performance.}

 \subsection{Simulation Results}
  The convergence performance of the proposed algorithms is illustrated in Fig. \ref{conv:simu_iter}, and we have a key observation. The algorithm takes around 20 iterations before convergence, with {\color{black}the termination criterion} $\epsilon = 10^{-3}$ under different settings of transmission power. Hence, the proposed algorithm in both blockage-free and blocked scenarios can rapidly converge.

\textbf{Blockage-free scenarios:} Simulations are conducted to illustrate the data rate performance of the HBF-proposed and other strategies without considering blockage events at first.

 Figure \ref{Fig_Rate_MR} shows the sum data rate performance of the five beamforming schemes under different numbers of MRs with the train velocity $v  = 360$ km/h and the transmission power $P_0 = 30$ dBm. As shown in the figure, we summarize the following four major observations. First, the HBF-Proposed algorithm achieves the highest sum rate, and the HBF-OMP performs slightly worse than the HBF-Proposed scheme due to the absence of the MMSE stage. Second, two ABF algorithms are constrained by the mutual interference of multi-user/streams, and there is apparent degradation of the ABF-AoD/AoA and ABF-Codebook algorithms. Both HBF-Proposed and HBF-OMP schemes outperform the ABF ones with the assistance of the digital precoder, which can facilitate interference mitigation. Third, the HBF-Benchmark algorithm performs the worst due to the randomization. Finally, it is observed that the system sum rate of all beamforming schemes rises with the increase of the number of MRs, by introducing the multiplexing gain and boosting the spectrum efficiency.

  Figure \ref{Fig_Rate_Speed} shows the impact of the train velocity on the system sum rate with 6 MRs deployed on the train. We can observe that the HBF-Proposed algorithm outperforms other schemes, while with the growth of the train speed, the gap between the HBF-Proposed algorithm and the HBF-OMP {\color{black}is narrowing} because of the severe ICI. When the velocity reaches up to 500 km/h, the gap between the two HBF algorithms or the two ABF algorithms becomes nearly indistinguishable, due to the Doppler effect.

\begin{figure*}
\centering
\begin{minipage}[t]{0.49\textwidth}
  \centering
  \includegraphics[width=0.75\columnwidth,height=53mm]{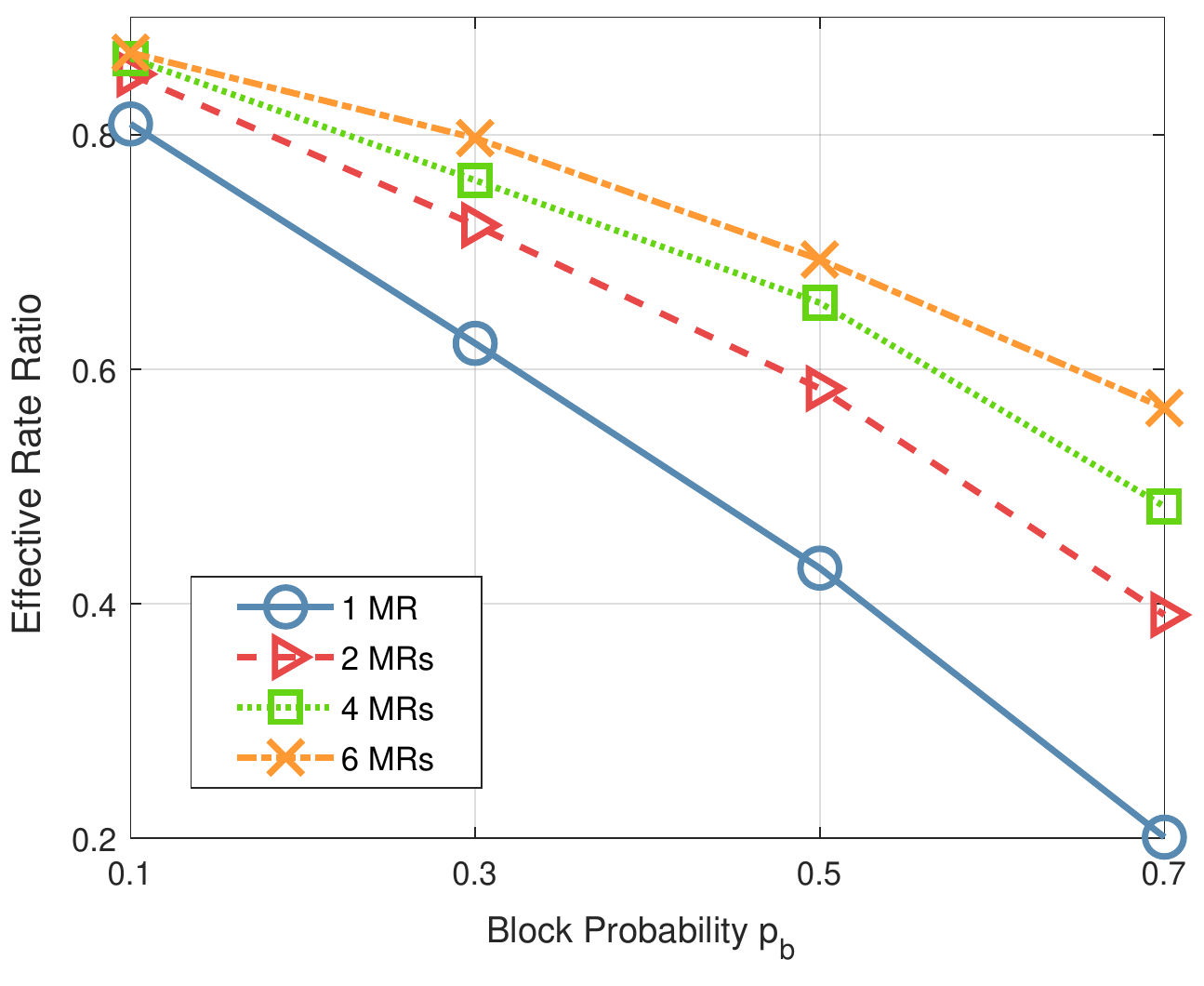}
  \caption{Effective rate ratio vs. blockage probability under different numbers of MRs.}
    \label{Fig_Ratio_MR_Block}
\end{minipage}
\hspace{0.05in}
\begin{minipage}[t]{0.49\textwidth}
  \centering
  \includegraphics[width=0.73\columnwidth]{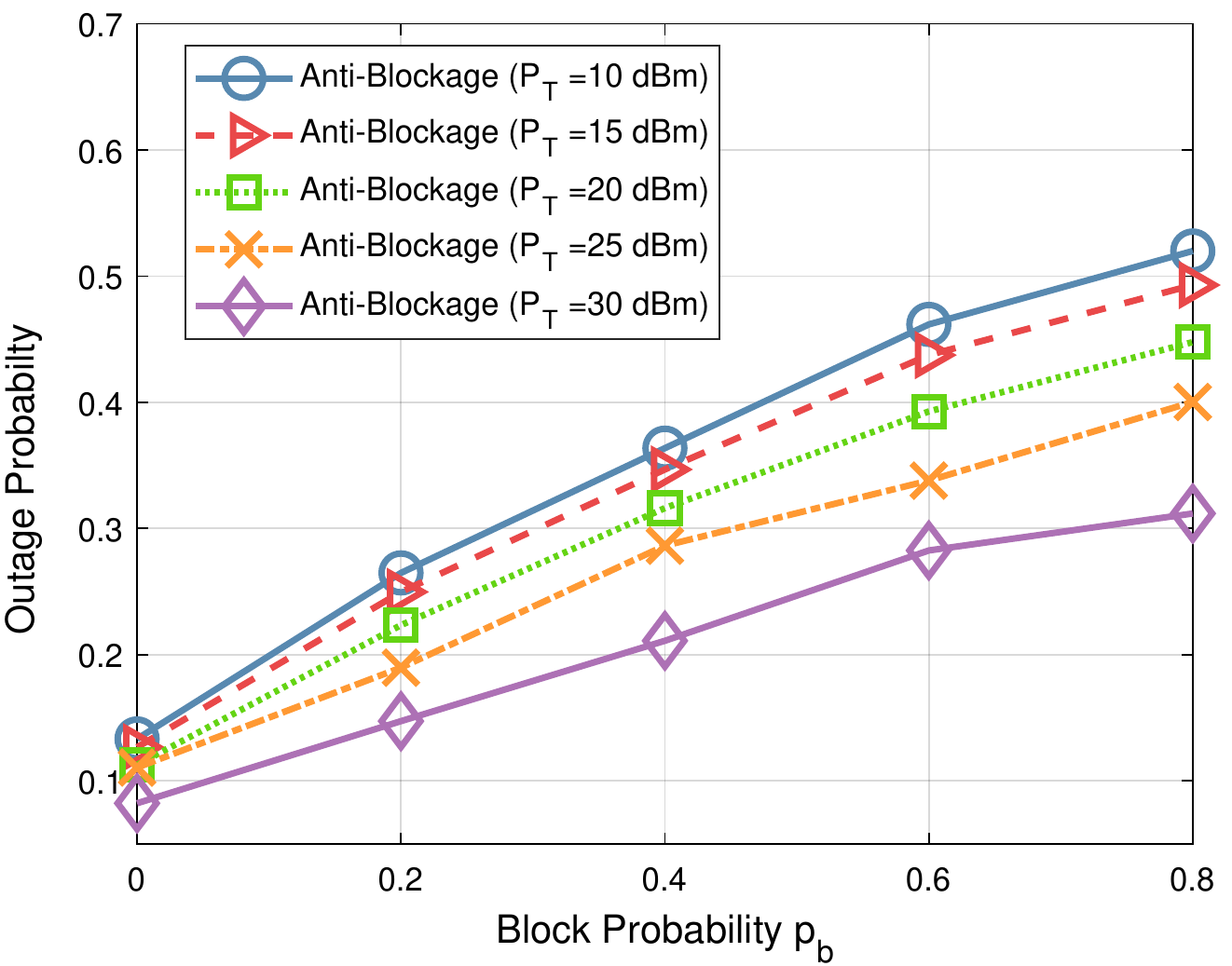}
  \caption{Outage probability vs. blockage probability under different transmission power.}
  \label{Fig_Out_Block}
\end{minipage}
\end{figure*}

\textbf{Blocked scenarios:}
 Next, we investigate the performance of the proposed anti-blockage beamforming algorithm, {\color{black}in terms of} various levels of blockage. When detected with severe attenuation of the observed signals at the MRs, the proposed anti-blockage beamforming algorithm turns to either the OMP stage update or both two-stage update. In the simulations, the blockage probability indicates each link may be blocked at an average possibility of $p_b$ \cite{Yang2015}.

 The instantaneous data rates with varying blockage probabilities and different numbers of MRs are shown in Fig. \ref{Fig_Rate_Block}. On one hand, higher blockage probabilities lead to severe degradation of the data rate with different numbers of MRs. On the other hand, when the blockage probability is low, the impact on the practical data rate is marginal. As blockage probability grows, {\color{black}random obstacles around the radio links }
 between the BS and MRs increase, {\color{black}thus leading} to degraded performance.

 Furthermore, we study the impact of different blockage probabilities under different power budget, as shown in Fig. \ref{Fig_Ratio_Cases_Block}. The train velocity is set to be 360 km/h. From Fig. \ref{Fig_Ratio_Cases_Block}, the proposed anti-blockage beamforming algorithm shows significantly better performance, comparing to the HBF-Proposed algorithm (labeled as ``When blocked'') which suffers from great blockage attenuation. Moreover, sufficient transmission power is observed to be more adaptive to blockage appearance and more robust against the link blockage. Besides, the proposed beamforming algorithm can improve the effective rate ratio {\color{black}by} around 20$\%$ compared with Algorithm 1 in severely-blocked scenarios with $p_b = 0.7$.

  To reveal the impact of concurrent transmission and spatial multiplexing, we set the number of MRs on the train to be 1, 2, 4 and 6, as shown in Fig. \ref{Fig_Ratio_MR_Block}. The transmission power is set to be 30 dBm. First, we can observe that more MRs deployed on the train help relieve the blockage attenuation significantly. Specifically, when the blockage probability reaches up to 0.7, the single MR case achieves only 20$\%$ effective rate ratio, which is only around 1/3 efficiency of the 6-MRs case. Second, the deployment of more MRs provides significant improvement in sum rate, especially in the severely-blocked scenario. Third, it is illustrated that higher blockage probability results in more link interruption thus degrades practical sum rates drastically.

 Figure \ref{Fig_Out_Block} shows the outage probability of the proposed anti-blockage algorithm, calculated by Eq. (\ref {p_out}) under different power budgets with varying blockage probabilities. The increasing blockage probability leads to greater outage possibility, as more dominant paths between transceivers are possibly obstructed. Moreover, the ascending outage probability with lower transmission power also fits well with the observations from Fig. \ref{Fig_Ratio_Cases_Block}.

\section{Conclusion}
 In this paper, we have investigated the efficient HBF design for mmWave HSR communications systems in both blockage-free and blocked scenarios, which is of great importance to provide gigabit data rates in future HSR communications. We first proposed a two-stage HBF algorithm in blockage-free scenarios.
 To combat random blockages, we have further proposed an anti-blockage strategy by intelligently invoking the proposed hybrid beamforming design to maintain system performance when the blockage is detected. The advantages of both algorithms have been verified via extensive simulations.
  For the future work, we will investigate the multi-cell coordinated beamforming in HSR communications.

\ifCLASSOPTIONcompsoc
\else
\fi

{\smaller{}
\bibliographystyle{IEEEtran}


\bibliography{hbf}}

\end{document}